\documentclass[
pra, twocolumn,
superscriptaddress,
 amsmath,amssymb,
 aps, 
floatfix
]{revtex4-2}

\usepackage{times}

\usepackage{amsfonts}
\usepackage{amssymb}
\usepackage{graphicx}
\usepackage{bbold}

\usepackage{accents}
\usepackage{graphicx}% Include figure files
\usepackage{gensymb} % for degree sign
\usepackage{dcolumn}% Align table columns on decimal point
\usepackage{bm}% bold math
\usepackage{amsmath}
\usepackage{graphicx}
\usepackage{yfonts}
\usepackage{float}
\usepackage{mathtools}
\usepackage{comment}
\DeclareMathOperator{\tr}{Tr}

\usepackage[utf8]{inputenc}
\usepackage{amsthm}
\newtheorem{theorem}{Theorem}

\newtheorem{corollary}{Corollary}[theorem]
\newtheorem{lemma}{Lemma}
\newtheorem{definition}{Definition}

\usepackage{notes2bib}
\bibnotesetup{
note-name = ,
use-sort-key = false
}
\renewcommand{\thesubsection}{\Alph{subsection}}
\newcommand{\bra}[1]{\left\langle #1 \right|}
\newcommand{\ket}[1]{\left| #1 \right\rangle}

\newcommand{\ketbra}[2]{\left|#1\middle\rangle\middle\langle#2\right|}

\usepackage[normalem]{ulem}
\usepackage{color}

\usepackage {hyperref}
\hypersetup{
     colorlinks   = true,
     linkcolor    = blue,
     citecolor    = blue,
     urlcolor     = blue,
     }
\makeatother
\begin{document} 
%\quickwordcount{main}
\preprint{APS/123-QED}

\title{Classical communication through quantum causal structures}%
\author{Kaumudibikash Goswami}

\email{k.goswami@uq.edu.au}
\author{Fabio Costa}%
\email{f.costa@uq.edu.au}
\affiliation{ 
 Australian Research Council Centre of Excellence for Engineered Quantum Systems, School of Mathematics and Physics,\\ University of Queensland, QLD 4072 Australia
}%
\begin{abstract}
Quantum mechanics allows operations to be in indefinite causal order. Recently there have been active discussions on enhanced communication strategies through exotic causal structures. In light of this, through the process matrix formalism, we formulate different classical capacities for a bi-partite quantum process. We find that a one-way communication protocol through an arbitrary process cannot outperform a causally separable process, i.e., we can send at most one bit per qubit. Next, we study bi-directional communication through a causally separable process. Our result shows, a bi-directional protocol cannot exceed the limit of one way communication protocol. Finally, we generalise this result to multi-party broadcast communication protocol through a definite ordered process.
\end{abstract}
\maketitle

\section{Introduction}
One of the key questions in quantum information is the rate at which a quantum channel can transmit classical information, as quantified by the \emph{classical capacity} of the channel \cite{Bennet98, Wilde13}. Holevo's seminal result \cite{Holevo73}, and following work \cite{Schumacher97, Holevo98}, provide upper bounds on the classical capacity, showing that each qubit can communicate at most one bit of classical information.

In a typical quantum communication protocol, the parties act in a fixed order. However, more general situations are possible, where causal order might be uncertain or even not defined. A practical example can be a distributed system, such as the internet, where different nodes communicate with each other. In such systems, local clocks can suffer from random errors and delays, leading to uncertainty in the ordering of the events \cite{Lamport78}. Even more radically, recent developments have shown possibilities of \emph{indefinite} causal structure, i.e., scenarios where the lack of order between the parties cannot be reduced to classical ignorance \cite{chiribella09, oreshkov12}. From a foundational point of view, this is relevant, for example, in quantum gravity scenarios, where quantum superposition of spacetimes can result in an indefinite causal order of events \cite{hardy2007towards,Hardy_2009_gravity_computer, zych2017bell}. Pragmatically, quantum control of causal order has been proposed as a resource for computation and communication \cite{chiribella09, araujo14, Guerin2016,Ara_jo_2017,Ebler_2018,chiribella2018indefinite,salek2018quantum, Ding2019, Caleffi_2020}, with several experimental implementations already performed \cite{Procopio_2015,rubino2017, Goswami_2018,goswami2018communicating, Wei_2019,Guo_2020,taddei2020experimental}.

In light of the foundational and applied relevance, it is important to understand how general quantum causal structure affects classical communication. In particular, one may wonder whether an indefinite causal structure can augment the classical communication capacity and possibly exceed the Holevo bound \cite{Holevo98,Schumacher97}. However, despite much work on various communication protocols, the notion of classical capacities in situations where the communicating parties themselves are indefinitely ordered has not yet been developed.

We address this gap through the process matrix formalism \cite{oreshkov12}. We develop expressions for the asymptotic capacity of a process, under different encoding and decoding settings, reducing to analogue expressions for quantum channels. We explore one-way communication protocols through an arbitrary process and show that such scenarios cannot exceed communication in definite causal order, i.e., we can send at most one bit per exchanged qubit. We also explore two-party communication protocols when causal order is definite but unknown (probabilistic). In such situations, the total bi-directional communication cannot exceed the maximum one-way communication---again, at most one bit per qubit in either direction. This extends to a similar bound for communication between multiple parties in a definite (but possibly probabilistic) causal order.

We present the work following way. In Section~\ref{sec:2}, we give an introduction to classical communication through quantum channels. In Section~\ref{sec:3}, we introduce the framework of the process matrix. In Section~\ref{sec:4}, we introduce asymptotic setting for processes, subsequently we define classical capacities of a process, and developed a bound for one way communication. In Section~\ref{sec:5} we develop a bound for  bi-directional communication protocol. We then generalise the bound for a multi-party broadcast communication protocol.

\section{ Classical communication through a quantum channel}
\label{sec:2}

Let us first review how one can use ordinary quantum channels to send classical information \cite{Wilde13, Gyongyosi2018}. In a one-way communication protocol, Alice has a classical message $m$, prepared according to some probability distribution $P(m)$, and encodes it into a quantum state $\rho_m$. She then sends it to Bob through a noisy quantum channel $\mathcal{N}$. Upon receiving the state, Bob extracts the message by using a positive operator valued measure (POVM) $\left\{E_{m'}\right\}_{m'}$, where $E_{m'}\geq 0$, $\sum_{m'} E_{m'}=\mathbb{1}$.  Here ${m'}$ denotes the measurement outcome. The conditional probability of Bob receiving a message $m'$ given that Alice sends the message $m$ is

\begin{align}
    p(m'|m) = \tr \left[ E_{m'}\mathcal{N}\left(\rho_m\right)\right].
    \label{Born}
\end{align}
The probability of error for a particular message $m$ is
\begin{align}
p_e(m) = p({m'} {\neq} m{|}m)  = 1 -p(m| m). \end{align}
The goal of the protocol is to minimise the maximal probability of error $p^*_e:=\max_{m} p_e(m)$. An \emph{asymptotic setting} is a scenario where Alice can use $n$ copies of the channel to send a $k$-bit message $m\in\{0,1\}^k$, where both $k$ and $n$ can be arbitrarily large. In other words, she encodes $k$ bits into an $n\geq k$ -bit message $X_m^{(n)}\in\{0,1\}^n$ and subsequently an $n$-qubit state $\rho_m^{(n)}$ and then sends each qubit through an independent copy of the channel. The classical capacity of the quantum channel $\mathcal{N}$, is defined as the maximal rate $C = k/n$ such that asymptotically, $n \rightarrow{\infty}$, one can achieve noiseless communication, $p^*_e\rightarrow{0}$ \cite{Shannon48, Wilde13}.

Different encoding and decoding strategies can lead to different asymptotic settings resulting in different classical communication capacities for a channel $\mathcal{N}$, which we review below.
 A quantification of classical communication possible through a channel $\mathcal{N}$ is given by the \emph{Holevo quantity} \cite{Holevo73}, defined as 
 
\begin{align} \label{holevochi}
  \chi(\mathcal{N}) :=  \underset{p(m), \rho_m}{{\mathrm{max}}} \ &S\left(\sum p(m) \mathcal{N}( \rho_m)\right) \nonumber \\ 
  &- \sum p(m)S\left(\mathcal{N}(\rho_m)\right).
\end{align}
Here $S(.)$ is the von Neumann entropy. Having introduced the Holevo quantity for a channel, it is interesting to see how this quantity is related to different classical capacities corresponding to different asymptotic configurations of channels. We discuss it below.

 \emph{Product encoding - Product decoding: } When the input quantum state is a product state of the form $\rho_{m}^{(n)} {=} \otimes_{i=1}^n \rho^{(i)}_m$ and the measurement operation is $E_{m'}^{(n)} {=} \otimes_{i=1}^n E_{m'}^{(i)}$ with each $E_{m'}^{(i)}$ acting on the qubit $\mathcal{N}(\rho^{(i)}_m)$. Let us consider, the measurement result produces an $n$-bit string $Y_{m'}^{(n)} \in \{0,1\}^n$ corresponding to the message $m'$. In the asymptotic setting, $n{\rightarrow}{\infty}$, the capacity in this setup is given by the conventional definition of classical capacity obtained by maximising the regularised mutual information, $I(Y_{m'}^{(n)}{:}X _m^{(n)})/n$, between Alice's input and Bob's output over the input probability distribution, the encoded quantum states and decoding measurement operators. The central idea of Shannon's capacity formula is that the mutual information $I(Y_{m'}^{(n)}{:}X _m^{(n)})$ is additive. Thus, the corresponding capacity is called \emph{one-shot capacity}, $C^{(1)}(\mathcal{N})$. This capacity is determined by the single use of the channel $\mathcal{N}$, with the optimised mutual information $I(Y_{m'}^{(1)}{:}X _m^{(1)})$, corresponding to the single-copy input and output variables $X _m^{(1)}$ and $Y_{m'}^{(1)}$ respectively, i.e.

\begin{align}
    C^{(1)}(\mathcal{N})&= \lim_{n{\rightarrow}{\infty}}\underset{p(m), \rho_m^{(n)}, E_{m'}^{(n)}}{{\mathrm{max}}}  \frac{I(Y_{m'}^{(n)}{:}X _m^{(n)})}{n} \nonumber \\
&=\underset{p(m), \rho_m^{(1)}, E_{m'}^{(1)}}{{\mathrm{max}}}I(Y_{m'}^{(1)}{:}X_{m}^{(1)} ).
    \label{accessible_information}
\end{align}
 \emph{Holevo's theorem} \cite{Holevo73} states that this quantity is upper bounded by
\begin{equation} 
    C^{(1)}(\mathcal{N}) \leq \chi(\mathcal{N}),
    \label{Eq:one-shot_and_holevo}
\end{equation}
with the $\chi(\mathcal{N})$ defined in Eq.~\eqref{holevochi}. Hereinafter, for the sake of clarity, we are going to represent $I(Y_{m'}^{(1)}{:}X_{m}^{(1)} )$ with $I({m'}{:}{m})$. 

 \emph{Product encoding - Joint decoding: } The difference with the previous case is that Bob can perform a joint measurement on the $n$-qubit system. The \emph{Holevo-Schumacher-Westmoreland (HSW) capacity}, $C(\mathcal{N})$, associated with this strategy turns out to be greater than one shot capacity, i.e. $C(\mathcal{N}){\geq} C^{(1)}(\mathcal{N})$, an effect known as \emph{super-additivity}. The HSW capacity is simply equal to the Holevo quantity \cite{Holevo98, Schumacher97}:
\begin{align}
    C(\mathcal{N}) {=} \chi(\mathcal{N}),
    \label{Holevo_capacity}
\end{align}
where $S(.)$ denotes the von Neumann entropy.

 \emph{Joint encoding - Joint decoding: } Here, Alice uses an entangled $n$-qubit state to encode the information and Bob performs a joint measurement on his output. The capacity associated with this strategy is given by \emph{regularised Holevo quantity} \cite{Holevo73,Holevo98,Schumacher97}: 
\begin{align}
C^E(\mathcal{N}) {=} \chi_\mathrm{reg}(\mathcal{N}) ,
\label{Eq:joint_encoding_capacity_channel}
\end{align}
with
\begin{align}
\chi_\mathrm{reg}(\mathcal{N}) {=} \lim_{n\to\infty} \frac{\chi(\mathcal{N}^{\otimes n})}{n}.    
\label{Regularised_Holevo}
\end{align}
It has been shown in \cite{Hastings_2009} that the capacity $C^E(\mathcal{N})$ in this case can be even greater than the HSW capacity --- $C^E(\mathcal{N}) \geq C(\mathcal{N})\geq C^{(1)}(\mathcal{N})$.

The Holevo quantity $\chi(\mathcal{N})$, and consequently the regularised Holevo quantity $\chi_\mathrm{reg}(\mathcal{N})$, are further upper bounded by $\log (d)$, where $d$ is the output dimension of the channel $\mathcal{N}$. With this we summarise a sequence of inequalities:
\begin{align}
    &I(m'{:}m) \leq C^{(1)}(\mathcal{N}) \leq C(\mathcal{N}) = \chi({\mathcal{N}}) \nonumber \\
   & \leq C^E(\mathcal{N}) = \chi_\mathrm{reg}(\mathcal{N}) \leq \log (d),
\end{align}
with $I(m'{:}m)$ being the unoptimized mutual information between Alice's input $m$ and Bob's output $m'$. A consequence of this chain of inequalities is that, for any communication setting, a $d$ dimensional quantum channel cannot transfer more than $\log (d)$ bits. In other words, quantum systems can carry at most one bit per qubit.

\section{The Process framework} \label{sec:3}
In conventional quantum communication protocols, the communicating parties act in a well defined order. However, quantum mechanics allows possibilities, where the order between the communicating parties is unknown or even indefinite \cite{hardy2007towards, chiribella09}. This possibility can be modelled within the so-called process matrix formalism \cite{oreshkov12, araujo15, oreshkov15}.
Consider a situation involving two parties --- Alice and Bob, each acting in a local laboratory. In each run of the experiment, each of them receives a quantum system in their respective laboratories, performs some operation on it and sends it out \cite{Araujo2017purification}; Alice's (Bob's) input and output systems will be denoted by $A_I$ ($B_I$) and $A_O$ ($B_O$), respectively. Each party can also access an additional system to perform their local operations. The most general operation is, therefore, a completely positive (CP) map $\mathcal{M}:X_I{\otimes} X_I'{\rightarrow} X_O{\otimes} X_O'$, where, for $X=A, B$; $X_I', X_O'$ denote the additional system and we use the system's label to represent the corresponding state space. 

It is convenient to represent CP maps as positive semidefinite matrices, $M {\in} X_I{\otimes} X_I'{\otimes} X_O{\otimes} X_O'$, using the Choi isomorphism \cite{choi75}:
\begin{multline}
    M^{X_IX_I'X_O X'_O} \\
    = \sum_{i,j =1}^{d_{X_I}d_{X'_I}}\ketbra{i}{j}^{X_IX_I'}\otimes\mathcal{M}(\ketbra{i}{j}^{X_IX_I'}). 
    \label{CP_map}
\end{multline}
Here, the set $\{\ket{i}\}$ represents an orthonormal basis in $X_I\otimes X_I'$ and $d_{X}$ represents the dimension of $X$.  If the map $\mathcal{M}$ is completely positive and trace preserving (CPTP), then the Choi representation gives an additional constraint  
\begin{equation}
    \tr_{X_OX'_O} M^{X_IX_I'X_O X'_O}  = \mathbb{1}^{X_{I}X_I'}.
\end{equation}
A CPTP map (also called \emph{channel}) represents an operation that can be performed with probability one, while a CP, trace non-increasing map generally is the conditional transformation corresponding to a particular outcome of a measurement.

The resource connecting the two communicating parties is described by the \emph{process matrix} $W^{A_IA_OB_IB_O}$. This encodes the background process that governs how the systems on which the parties act relate to each other, be it a shared state, a channel from one to the other, or more general scenarios. The process matrix $W$ has to satisfy a set of constraints:
\begin{align}
    W &\geq 0, \\
    \tr W &= d_{A_O}d_{B_O}, \\
    _{B_IB_O}W &= _{A_OB_IB_O}W, \\ \label{AtoBstate}
    _{A_IA_O}W &= _{B_OA_IA_O}W,\\ \label{noloops}
    W &= _{A_O}W + _{B_O}W - _{A_OB_O}W. 
\end{align}
Here, ${_x}W := \mathbb{1}^x/d_x \otimes Tr_{x}W$ is the `trace-and-replace' notation \cite{araujo15}, which discards subsystem $x$ and replaces it with the normalised identity. Here appropriate reordering of the tensor factors is implied. The concatenation of Alice's (Bob's) local operation $\mathcal{M}^A$ and $\mathcal{M}^B$ with the process $W$ is given by $W {*} M^ A {*} M^B$ where $M^A$ ($M^B$) is the Choi representation of the map $\mathcal{M}^A$ ($\mathcal{M}^B$) and `*' is the link product \cite{chiribella12} defined as

\begin{align}
    P * Q := \tr _{\mathcal{P} \cap \mathcal{Q}} [(\mathbb{1}^{\mathcal{P} \setminus \mathcal{Q}} \otimes P^{T_{\mathcal{P} \cap \mathcal{Q}}})(Q \otimes \mathbb{1}^{\mathcal{Q} \setminus \mathcal{P}})]. 
    \label{link-product}
\end{align}
Here, $\mathcal{P}$ and $\mathcal{Q}$ are the Hilbert spaces associated with $P$ and $Q$, the superscript `$T_{\mathcal{P} \cap \mathcal{Q}}$' represents partial transpose on the shared Hilbert spaces.

\begin{figure}[!t]
\begin{center}
%\vspace{-5mm}
\includegraphics[width=0.8\columnwidth]{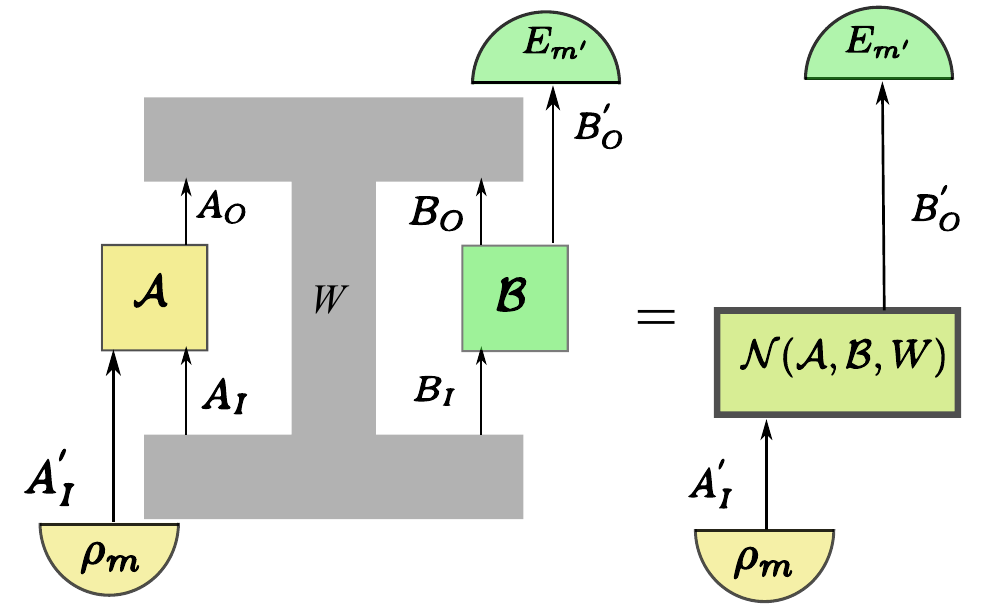}
\vspace{-2mm}
\caption{ A process $W^{A_IA_OB_IB_O}$ with two CPTP maps $A^{A'_IA_IA_O}$ and $B^{B_IB_OB'_O}$ forms a new channel $\mathcal{N}(\mathcal{A},\mathcal{B},W)$, as in Eq.~\eqref{channel}, with input system $A'_I$ and the output system $B'_O$. Alice can use this channel to communicate to Bob by encoding the quantum state $\rho_m$ at her input system and Bob performing a POVM measurement $E_{m'}$ at his output system.}
\label{fig:Protocol}
\end{center}
\end{figure}

\section{One directional communication through an indefinitely ordered process}\label{sec:4}

In this section we introduce our classical communication protocol through an arbitrary process as shown in Fig.~\ref{fig:Protocol}. In this protocol, both Alice and Bob can use some quantum channels $\mathcal{A}$ and $\mathcal{B}$. Alice's channel $\mathcal{A}$ has $A'_I A_I$ as input and $A_O$ as output, while Bob's channel $\mathcal{B}$ has $B_I$ as input and $B_O B'_O$ as output.

The process matrix $W^{A_IA_OB_IB_O}$ acting on these channels forms a new quantum channel $\mathcal{N}(\mathcal{A}, \mathcal{B}, W)$ with input quantum system $A'_I$ and the output quantum system $B'_O$, as shown in Fig.~\ref{fig:Protocol}. The Choi representation of this new channel is 

\begin{align}
    N(\mathcal{A}, \mathcal{B},& W)^{A'_IB'_O}  \nonumber \\
    :=& W^{A_IA_OB_IB_O} * (A^{A'_IA_IA_O} * B^{B_IB_OB'_O}) \nonumber \\
    =& W^{A_IA_OB_IB_O} * \big(A^{A'_IA_IA_O} \otimes B^{B_IB_OB'_O}\big).
    \label{channel}
\end{align}
Here, $A$ and $B$ are the Choi representations of the quantum channels $\mathcal{A}$ and $\mathcal{B}$ respectively.

\subsection{Asymptotic setting}

Similar to the conventional classical communication through quantum channels, we need to introduce a notion of asymptotic setting for process, namely to formalise the notion of repeated use of independent copies of a process. The goal turns out to be non-trivial as one can construct different asymptotic settings by allowing non-product channels across the different copies \cite{Jia_2018, Gu_rin_2019}, resulting in non-trivial constraints on the admissible operations and processes \cite{Perinotti2017, Kissinger2019}. For example, Alice could feed the output of her first channel to her second one. This, however, would require extra knowledge about the causal relations between the different uses of the process and, for a process with bidirectional signalling, it would be incompatible with Bob sending his second output to the first input. As we are investigating causal structures as communication resources, we assume that all available causal relations are encoded in the process itself, which leads to the asymptotic setting, first introduced in Ref.~\cite{Jia_2019}, where only product operations across different parties are allowed.

 Our choice of asymptotic setting results in a set of independent channels $\mathcal{N}_j = \mathcal{N}(\mathcal{A}_j, \mathcal{B}_j, W)$,  as shown in Fig.~\ref{fig:process_asymptotic}. Here $\mathcal{A}_j$, $\mathcal{B}_j$  are the local operations performed by Alice and Bob respectively. Alice encodes her message $m$ in a quantum state $\rho_m^{(n)}{\in} \otimes_{j=1}^n A_I^{'j}$ and sends the state to Bob through the channels $\{\mathcal{N}_j\}$. After receiving the transformed state, Bob performs a POVM on his quantum system $\otimes_{j=1}^n B^{'j}_O$. 
 With this, we conceptualise a protocol for one way communication from Alice to Bob ($A {\rightarrow}{B}$) in the following way:

\begin{definition}
\label{Def:A_to_B_code}
Given a bipartite processes matrix $W^{AB}$, we define an \emph{$A{\rightarrow} B$ protocol with $n$ uses of $W$} as
\begin{enumerate}
    \item A set of local operations $\{\mathcal{A}_j,\mathcal{B}_j\}_{j=1}^n$, where
    \begin{align*}
        \mathcal{A}_j &: A^{'j}_{I}\otimes A^j_{I} {\rightarrow} A^j_{O},  \\
        \mathcal{B}_j &: B^j_{I} \rightarrow B^j_{O}\otimes B^{'j}_{O}
    \end{align*}
    are CPTP maps;
    \item A state encoding $m\mapsto \rho^{(n)}_m \in \bigotimes_{j=1}^n A^{'j}_{I}$, where $\rho^{(n)}_m\geq 0$ and $\tr \rho^{(n)}_m =1$;
    \item A decoding POVM $\{E^{(n)}_{m'}\}_{m'}$, where $E^{(n)}_{m'}\geq 0$ and $\sum_{m'} E^{(n)}_{m'} = \mathbb{1}$.
\end{enumerate}
\end{definition}

Such a protocol produces a classical channel described by the conditional probabilities
\begin{equation}
    P(m'|m) = \tr \left[ E^{(n)}_{m'}\bigotimes_{j=1}^n \mathcal{N}_j\big(\rho^{(n)}_m\big) \right].
\end{equation}
We say that two protocols for the same process $W$ are \emph{equivalent} if they produce the same conditional probabilities $P(m'|m)$.

In general, the ancillary spaces $A^{'j}_{I}$, $B^{'j}_{O}$ need not be isomorphic for different $j$. However, we can always embed each of them into a space isomorphic to one of the highest dimension. In the following, we assume that all spaces are of equal dimension and are identified through a choice of canonical basis.

Note this specific arrangement of channels results in a non-stationary asymptotic setting. Formulating communication capacity of such a setup poses a non-trivial challenge \cite{Verdu_Han_94,Hayashi_non-stationary}. To alleviate this issue, we employ a scheme to make the channels stationary. Specifically, we replace the local operations $\{\mathcal{A}_j\}$ and $\{\mathcal{B}_j\}$ with fixed local operations $\{{\mathcal{A}}\}$ and $\{{\mathcal{B}}\}$ respectively with support of additional local CPTP maps $\mathcal{E}_j$ and $\mathcal{F}_j$, where $\mathcal{A}_j{=}{\mathcal{A}}{*}\mathcal{E}_j$ and $\mathcal{B}_j{=}\mathcal{F}_j{*}{\mathcal{B}}$. Thus we have multiple independent and identical copies of the channel ${\mathcal{N}}{=}\mathcal{N}({\mathcal{A}},{\mathcal{B}},W)$. Feasibility of this approach is due to the fact that, in an $A\rightarrow B$ protocol, Bob's output can be discarded, i.e., in such a protocol a process matrix $W$ can be replaced by $_{B_O}W$, as shown in Fig.~\ref{fig:simplify_lemma} and in Refs.~\cite{oreshkov12, chiribella08, chiribella09, morimae2014process}.
As we are going to use this fact multiple times, we formulate it as a lemma and prove it below for completeness:

\begin{lemma}\label{noBO}
If Alice has trivial ancillary output, $\mathcal{A}:A'_I\otimes A
_I\rightarrow A_O$, we can replace $W$ with $_{B_O}W$:
\begin{equation}
    \mathcal{N}(\mathcal{A},\mathcal{B}, W) = \mathcal{N}(\mathcal{A},\mathcal{B}, {}_{B_O}W).
\end{equation}
\end{lemma}
\begin{proof}
It is sufficient to show $A{*}W = A{*}{}_{B_O}W$. Using condition \eqref{noloops}, we can write $A{*}W = A{*}{}_{A_O}W + A{*}{}_{B_O}W - A{*}{}_{A_O B_O}W$. As the second and third terms are already in the desired form, we only need to look at the first term: 
\begin{multline}
A{*}{}_{A_O}W = \tr_{A_IA_O} \left[ A^{A'_IA_IA_O}\cdot \left( \frac{\mathbb{1}^{A_O}}{d^A_O} \otimes \tr_{A_O}W \right) \right]\\
= \frac{1}{d^A_O} \tr_{A_I} \left[\left(  \tr_{A_O}A^{A'_IA_IA_O} \right) \cdot \left(  \tr_{A_O}W \right) \right] \\
= \frac{1}{d^A_O} \tr_{A_I} \left[ \mathbb{1}^{A'_I A_I} \cdot \left(  \tr_{A_O}W \right) \right] \\
= \frac{\mathbb{1}^{A'_I}}{d^A_O} \otimes \tr_{A_IA_O} W = \frac{\mathbb{1}^{A'_I}}{d^A_O} \otimes \tr_{A_IA_O}\left( _{B_O}W \right),
\end{multline}
where we used $\tr_{A_O}A^{A'_IA_IA_O} = \mathbb{1}^{A'_I A_I}$ (because $A^{A'_IA_IA_O}$ is CPTP) in the third line and Eq.~\eqref{AtoBstate} in the last line.
\end{proof}

This lemma allows us to replace Bob's operation $\mathcal{B}_j$ by $\sigma^{B_O}\otimes \big(\tr_{B_O} \mathcal{B}_j\big)^{B_IB'_O}$, with $\sigma^{B_O}$ being an arbitrary state. For the encoding operation $\mathcal{A}_j$, on the other hand, we extend the input system to make it a controlled operation while treating the control state as Alice's extended encoded message. Thus we present the following theorem.

\begin{figure}[!t]
\begin{center}
%\vspace{-5mm}
\includegraphics[width=\columnwidth]{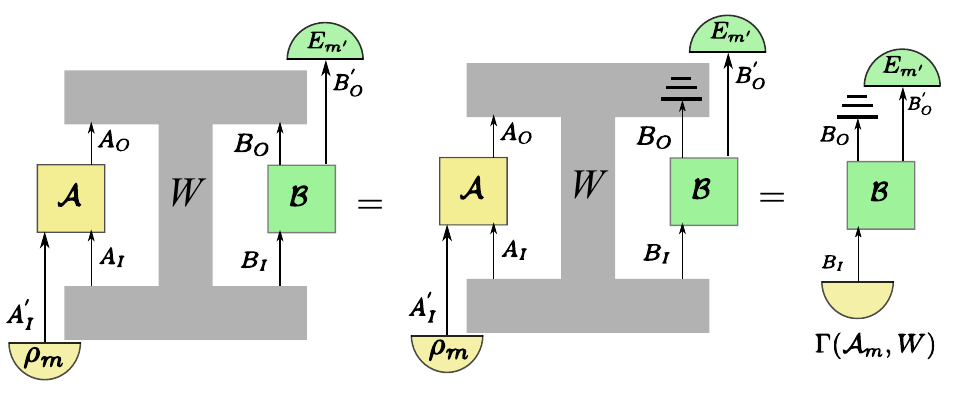}
\vspace{-2mm}
\caption{Pictorial depiction of Lemma~\ref{noBO}. A one way communication from Alice to Bob through a channel $\mathcal{N}(\mathcal{A},\mathcal{B},W)$ simplifies to a channel $\mathcal{N}(\mathcal{A},\mathcal{B}, _{B_O}W)$. The rightmost picture shows further simplification to a quantum state  $\Gamma(\mathcal{A}_m,W)$ with the system $B_O$ being discarded, as in Eq.~\eqref{simplification}.     \vspace{0 mm}}
\label{fig:simplify_lemma}
\end{center}
\end{figure}

\begin{theorem} \label{Theorem:asymptoticsemplification}
Every $A\rightarrow B$ protocol is equivalent to one with fixed local operations
\begin{align}
   \bar {\mathcal{A}} &: A^{''j}_{I}\otimes A^j_{I} \rightarrow A^j_{O},  \\
        \bar{\mathcal{B}} &: B^j_{I} \rightarrow B^j_{O}\otimes B^{''j}_{O},
\end{align}
state encoding 
\begin{equation}\label{newencoding}
    \bar{\rho}^{(n)}_m = \bigotimes_{j=1}^n \mathcal{E}_j\left(\rho^{(n)}_m\right),
\end{equation} 
and decoding POVM
\begin{equation}\label{newdecoding}
    \bar{E}^{(n)}_{m'} = \bigotimes_{j=1}^n \mathcal{F}_j^{\dag}(E^{(n)}_{m'}),
\end{equation}
 where $\mathcal{E}_j{:}A^{'j}_I {\rightarrow} A^{''j}_I$,  $\mathcal{F}_j{:}B^{'j}_O {\rightarrow} B^{''j}_O$ are CPTP and $\mathcal{F}^{\dag}$ denotes the Hilbert-Schimdt adjoint, defined through $\tr\left[ A^{\dag} \mathcal{F}(B) \right]= \tr \left[ \mathcal{F}^{\dag}\left( A^{\dag}\right) B \right]$.
 \end{theorem}

\begin{proof} (See Fig.~\ref{fig:simplify_theorem} for a pictorial representation of the proof.)
Let us start with the decoding.
Since Alice only performs CPTP maps with no ancillary output, we can apply Lemma \ref{noBO} and replace the process matrix $W$ with ${}_{B_O}W$, which is equal to identity on $B_O$. This implies that any $A\rightarrow B$ protocol is equivalent to one where we replace the local operations $\mathcal{B}_j$ with ${\sigma}^{B^j_O}\otimes \tr_{B^j_O}\mathcal{B}_j$ for some arbitrary state $\sigma$. Choosing the space $B^{''j}_{O}$ isomorphic to $B^j_{I}$, we see that the original combination of local operations $\mathcal{B}_j$ and decoding POVM is equivalent to performing the fixed operation $\bar{\mathcal{B}} = {\sigma}^{B^j_O} \otimes \mathcal{I}^{B^{j}_{I} \rightarrow B^{''j}_{O}}$ in each lab, and decoding POVM as in Eq.~\eqref{newdecoding}, with $\mathcal{F}_j = (\tr_{B^j_O} \mathcal{B}_j\circ \mathcal{I}^{B^{j}_{I} \rightarrow B^{''j}_{O}})^{B^{'j}_{O}}$.

Now for the encoding side: we set $A^{''j}_I = \mathcal{L}\left(\mathbb{C}^n_j\right)\otimes A^{'j}_I$ and define the controlled operation $\bar{\mathcal{A}}: \mathcal{L}\left(\mathbb{C}^n_j\right)\otimes A^{'j}_I \rightarrow A^j_{O}$ as
\begin{equation}
 \bar{\mathcal{A}}\left(\sigma\otimes \rho \right) = \sum_{j=1}^n \bra{j}\sigma \ket{j} \mathcal{A}_j\left(\rho\right)^{A^{j}_{O}},    
\end{equation}
which is manifestly CPTP. For canonical basis states in $\mathcal{L}\left(\mathbb{C}^n_j\right)$, this map gives 
\begin{equation}
    \bar{\mathcal{A}}\left(\ketbra{j}{j}\otimes \rho \right) = \mathcal{A}_j\left(\rho\right),     \label{Eq:A_bar}
\end{equation}
so the choice of local operation can be encoded into a choice of initial state, expanding the original encoding state as in Eq.~\eqref{newencoding}, with the maps $\mathcal{E}_j:A^{'j}_O \rightarrow A^{''j}_O$ defined as
\begin{equation}
    \mathcal{E}_j\left(\rho^{A^{'j}_I}\right) = \big(\ketbra{j}{j}\otimes \rho\big)^{A^{''j}_I}.
\end{equation}
\end{proof}

The relevance of this theorem is twofold. First, it shows that any protocol involving a different choice of local operations can be reproduced by fixing the local operations once and for all. This means that an asymptotic setting for processes can always be mapped to an asymptotic setting were the same channel is used $n$ times, $\mathcal{N}(\bar{\mathcal{A}},\bar{\mathcal{B}}, W)^{\otimes n}$. We call a protocol of this type \emph{stationary}. Second, state encoding and decoding POVM of an arbitrary protocol transform into those of a stationary one through product maps, Eqs.~\eqref{newencoding} and \eqref{newdecoding}. This means that the transformation preserves the nature of the asymptotic setting, viz.\ joint/product encoding or decoding. From now on, we will represent $\bar{\mathcal{A}}$, $\bar{\mathcal{B}}$, $\bar{\rho}_m^{(n)}$ and $\bar{E}_{m'}^{(n)}$ without the bar on top.

\begin{figure}[!t]
\begin{center}
%\vspace{-5mm}
\includegraphics[width=\columnwidth]{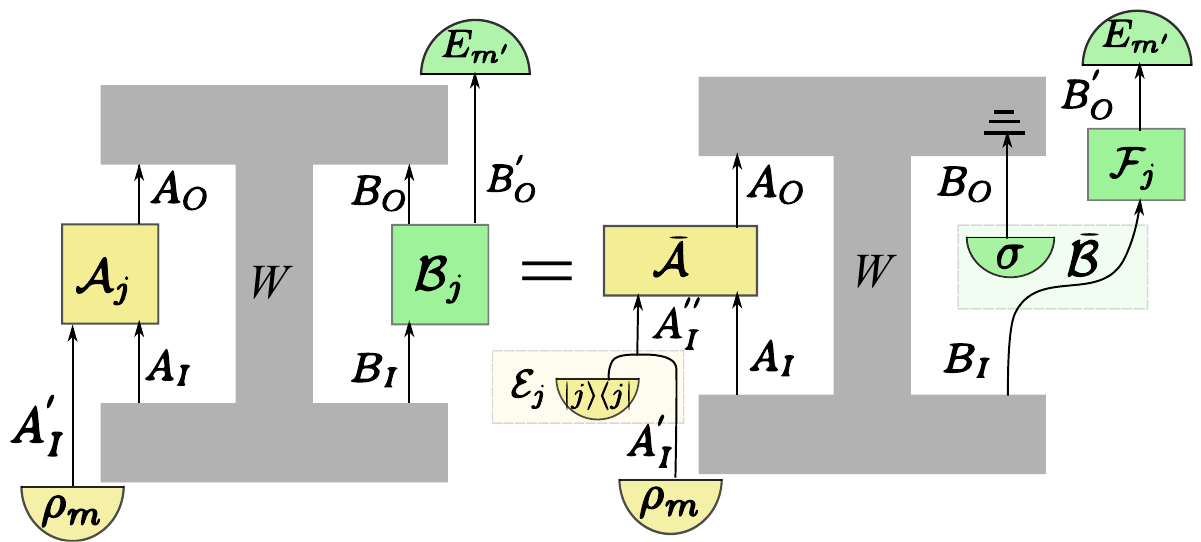}
\vspace{-2mm}
\caption{Pictorial depiction of Theorem~\ref{Theorem:asymptoticsemplification}. We convert a non-stationary channel $\mathcal{N}(\mathcal{A}_j,\mathcal{B}_j, W)$ to a stationary channel $\mathcal{N}(\bar{\mathcal{A}},\bar{\mathcal{B}}, W)$. Due to Lemma~\ref{noBO}, Bob's system $B_O$ can be set to a fixed state $\sigma$ and corresponding operation $\bar{\mathcal{B}} = {\sigma}^{B^j_O} \otimes \mathcal{I}^{B^{j}_{I} \rightarrow B^{''j}_{O}}$. Alice's operation, on the other hand, can be extended to a controlled CPTP map $\bar{\mathcal{A}}$ as described in Eq.~\eqref{Eq:A_bar}.\vspace{0 mm}}
\label{fig:simplify_theorem}
\end{center}
\end{figure}

\subsection{Classical capacities of a quantum process}

\emph{Holevo quantity for a process:} Having introduced a stationary protocol with an asymptotic setting of the channel $\mathcal{N}({\mathcal{A}},{\mathcal{B}},W)$, as shown in Theorem~\ref{Theorem:asymptoticsemplification},  we can define the corresponding Holevo quantity for a process $W$ as
\begin{align}
    \chi(W)^{A \rightarrow{B}}:= \underset{\mathcal{A}, \mathcal{B}}{{\mathrm{max}}}\ {\chi}\left[{\mathcal{N}(\mathcal{A}, \mathcal{B}, W)}\right].
    \label{holevo_w}
\end{align}
We also introduce the $n$-th extension $\chi(W^{\otimes n})^{A \rightarrow{B}}$ of the above quantity as 
\begin{align}
    \chi(W^{\otimes n})^{A \rightarrow{B}}:= \underset{\mathcal{A}, \mathcal{B}}{{\mathrm{max}}}\ {\chi}\left[{\mathcal{N}(\mathcal{A}, \mathcal{B}, W)^{\otimes n}}\right].
    \label{holevo_n_th_w}
\end{align}

\emph{Communication capacity for a process:} We can associate different communication capacities to an arbitrary process as
\begin{align}
C^{\sharp} (W)^{A\rightarrow{B}} = \max_{\mathcal{A}, \mathcal{B}}C^{\sharp}\left(\mathcal{N}(\mathcal{A}, \mathcal{B},W)\right).
\label{Eq:Capacity_w}
\end{align}
Where $C^{\sharp}=C^{(1)}, C, C^{\textrm{E}}$. Here $C^{(1)}(W)^{A\rightarrow{B}}$ represents product encoding-product decoding capacity, as in Eq.~\eqref{accessible_information}, $C(W)^{A\rightarrow{B}}$ represents product encoding-joint decoding capacity, as in Eq.~\eqref{Holevo_capacity} and finally, $C^{E}(W)^{A\rightarrow{B}}$ represents joint encoding-joint decoding capacity, as in Eq.~\eqref{Regularised_Holevo}. 

We can relate the Holevo quantity for a process to different $C^{\sharp} (W)^{A\rightarrow{B}}$ capacities. We show this in the following lemma.

\begin{lemma}
Different capacities associated with an arbitrary process $W$ are related to the Holevo quantity $\chi(W)^{A \rightarrow{B}}$ in the following way.

product encoding-product decoding: 
\begin{align}
   C^{(1)} (W)^{A\rightarrow{B}} \leq \chi(W)^{A \rightarrow{B}}.
\end{align}

product encoding-joint decoding:
\begin{align}
   C (W)^{A\rightarrow{B}} = \chi(W)^{A \rightarrow{B}}.
\end{align}

Joint encoding-joint decoding: 
\begin{align} \label{jointencoding}
   C^{E} (W)^{A\rightarrow{B}} = \lim_{n\rightarrow{\infty}}\frac{\chi(W^{\otimes n})}{n} .
\end{align}

\end{lemma}
  
\begin{proof}
product encoding-product decoding: 

Using Eqs.~\eqref{Eq:Capacity_w},~\eqref{Eq:one-shot_and_holevo} and ~\eqref{holevo_w}, we can write
\begin{align}
  C^{(1)} (W)^{A\rightarrow{B}} &= \max_{{\mathcal{A}},{\mathcal{B}}}C^{(1)} (\mathcal{N}({\mathcal{A}},{\mathcal{B}},W))  \nonumber \\
 & \leq \max_{{\mathcal{A}},{\mathcal{B}}} \chi(\mathcal{N}({\mathcal{A}},{\mathcal{B}},W))  \nonumber \\
 &=\chi(W)^{A \rightarrow{B}}.
 \label{Eq:one-shot}
\end{align}

product encoding-joint decoding: 

Using Eqs.~\eqref{Eq:Capacity_w}, ~\eqref{Holevo_capacity} and ~\eqref{holevo_w} we can write
\begin{align}
  C (W)^{A\rightarrow{B}} &= \max_{{\mathcal{A}},{\mathcal{B}}}C (\mathcal{N}({\mathcal{A}},{\mathcal{B}},W))  \nonumber \\
 & = \max_{{\mathcal{A}},{\mathcal{B}}} \chi(\mathcal{N}({\mathcal{A}},{\mathcal{B}},W))  \nonumber \\
 &=\chi(W)^{A \rightarrow{B}}.
 \label{Eq:product_enc_joint_decod}
\end{align}

 Joint encoding-joint decoding:  
 
 Using Eqs.~\eqref{Eq:Capacity_w}, ~\eqref{Eq:joint_encoding_capacity_channel}, ~\eqref{Regularised_Holevo} and ~\eqref{holevo_n_th_w} we can write

\begin{align}
  C^E (W)^{A\rightarrow{B}} &=\max_{{\mathcal{A}},{\mathcal{B}}}C^E (\mathcal{N}({\mathcal{A}},{\mathcal{B}},W)) \nonumber \\
  &=\max_{{\mathcal{A}},{\mathcal{B}}}\chi_{\mathrm{reg}} (\mathcal{N}({\mathcal{A}},{\mathcal{B}},W)) \nonumber \\
  &=\max_{{\mathcal{A}},{\mathcal{B}}} \lim_{n\rightarrow{\infty}}\frac{\chi\left[\mathcal{N}({\mathcal{A}},{\mathcal{B}},W)^{\otimes n}\right]}{n} \nonumber \\
  &=\lim_{n\rightarrow{\infty}}\frac{\chi(W^{\otimes{n}})}{n}.
  \label{proofjointencoding}
  \end{align}
\end{proof}

\subsection{Bounds on the classical capacities of a quantum process}
Although we have been able to reduce the classical capacities of processes to that of channels, our results so far do not provide an upper bound on how much information can be transmitted through a process. This is because the channel $\mathcal{N}(\mathcal{A}, \mathcal{B}, W)$ can have arbitrary input and output dimension.

To establish a bound, we first describe our protocol from a slightly different point of view. With Alice's input ensemble $\{p(m),{\rho_m}\}$, we can introduce a concatenation of ${\rho_m}$ with the channel ${\mathcal{A}}$ as $
{A}_m^{A_IA_O} = {\rho}_m^{A''_{I}} * {A}^{A''_IA_IA_O}$ with ${A}_m$ being the Choi representation of the resulting CPTP map ${\mathcal{A}}_m$. Similarly, we can combine Bob's channel ${\mathcal{B}}$ and POVM operation ${\{E_m\}}$ to describe a set of CP maps $\{{B}_{m'}^{B_IB_O} = B^{B_I B_O B''_O} * ( {E}_{m'}^T)^{B''_O}\}_{m'}$, where $\sum_{m'} {B}_{m'}^{B_IB_O}$ is a CPTP map and $({E}_{m'}^T)^{B''_O}$ is the Choi representation of ${E}_{m'}$. The superscript `$T$', denoting the transpose operator, is due to definition \eqref{CP_map}, according to which the Choi of a measurement operator ${E}_{m'}$ is its transpose ${E}_{m'}^T$. 

With this in mind, we present two theorems, that apply respectively to the product and joint encoding scenarios.

\begin{figure*}[!t]
\begin{center}
%\vspace{-5mm}
\includegraphics[width=2\columnwidth]{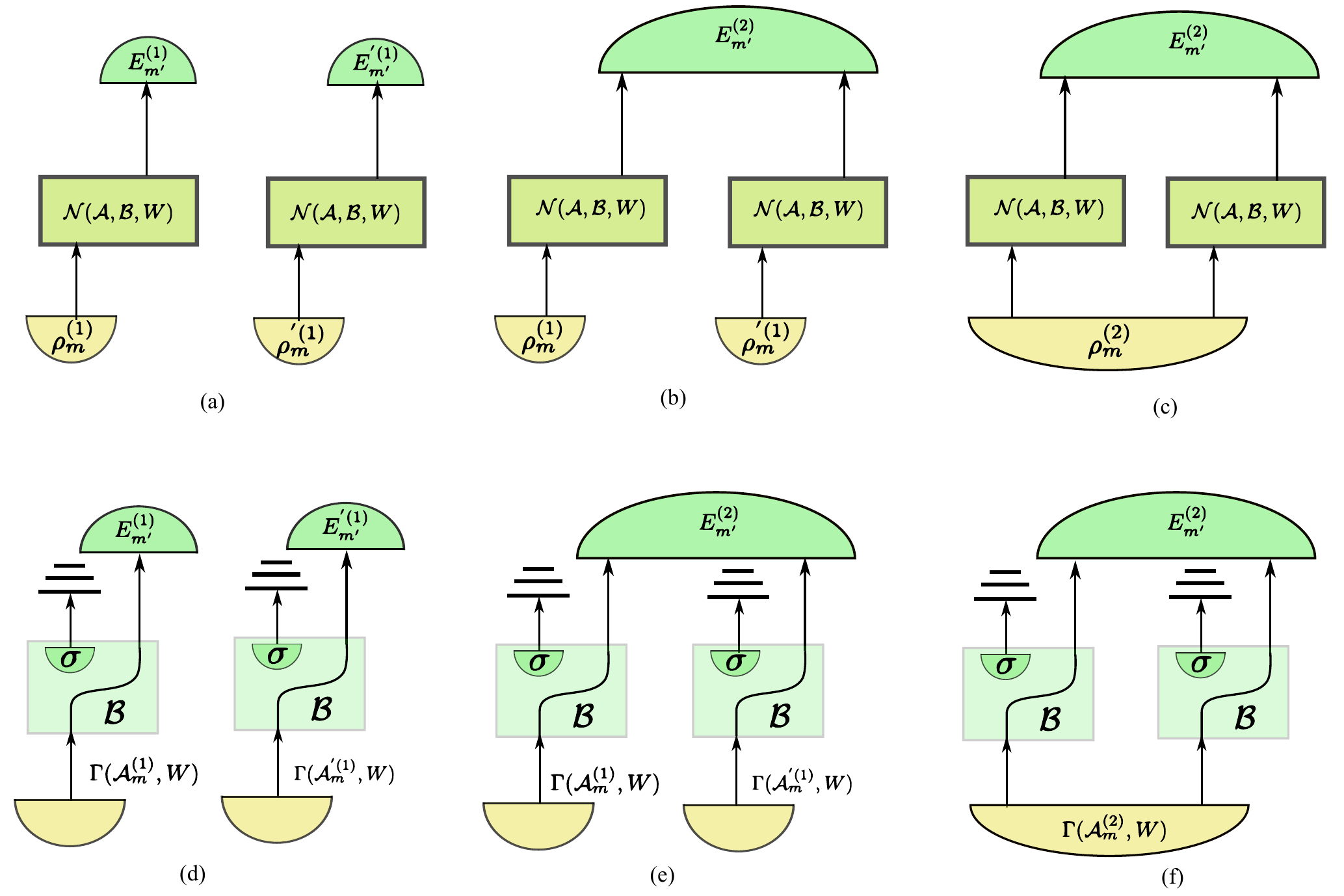}
\vspace{-2mm}
\caption{Encoding and decoding schemes. Here we show two copies of process $W$. Fig. (a), (b), (c) demonstrate product-encoding product decoding, product encoding joint decoding and joint encoding-joint decoding respectively. Product encoding is achieved using the joint state $\rho_m^{(1)}{\otimes}\rho_m^{'(1)}$ and joint encoding is achieved using the entangled state $\rho_m^{(2)}$. Similarly, we use $E_{m'}^{(1)}{\otimes}E_{m'}^{'(1)}$ for product decoding and $E_{m'}^{(2)}$ for joint decoding. Figs. (d), (e), (f) are the simplifications due to Theorem~\ref{Theorem:asymptoticsemplification} and Theorem~\ref{Theorem:2}. Relevant labelling of the Hilbert spaces are described in the text. }
\label{fig:process_asymptotic}
\end{center}
\end{figure*}

\begin{theorem}\label{Theorem:2}
In a one-way communication scenario, the optimisation of the Holevo quantity of a process $W$ can be simplified as
\begin{align}
    \chi(W)^{A \rightarrow{B}} = \underset{{\rho}_m,{\mathcal{A}}, p(m)}{{\mathrm{max}}}
    \ &S\left[\sum p(m) {\Gamma} ({A}_m, W )\right] - \nonumber \\
    &\sum p(m)S\left[{\Gamma} ({A}_m ,W)\right].
    \label{W_Holevo}
\end{align}
Here the $\Gamma(.)$ is a map that transforms the Choi representation of the CPTP operation ${\mathcal{A}}_m$ and the process matrix $W$, to a quantum state on Bob's input space $B_I$.
\end{theorem}  

\begin{proof}
The reduced process \cite{araujo15} on which Bob applies his CPTP map ${B}_{m'}^{B_IB_O}$ is described by ${A}_m^{A_IA_O} * W^{A_IA_OB_IB_O}$. Now, using Lemma \ref{noBO}, we can write
\begin{align}
     {A}_m^{A_IA_O} {*} W^{A_IA_OB_IB_O} &{=} \mathbb{1}^{B_O} {\otimes} \left({A}_m^{A_IA_O} {*} \frac{\tr_{B_O}W}{d_{B_O}}^{A_IA_OB_I}\right) \nonumber \\
     &= \mathbb{1}^{B_O} {\otimes} {\Gamma} ({A}_m,W )^{B_I}.
     \label{simplification}
\end{align}
Where ${\Gamma} ({A}_m, W ) {:=} {A}_m {*} \tr_{B_O}W/d_{B_O}$. In other words, as shown in Fig. \ref{fig:simplify_lemma}, we can simplify Alice's CPTP map and the process to a quantum state ${\Gamma} ({A}_m, W )$ in the Hilbert space $B_I$, with Bob's output at $B_O$ being discarded. The maximum classical information that can be encoded in the ensemble $\{p(m), {\Gamma} ({A}_m, W )\}$ is given by Eq.~\eqref{W_Holevo} \cite{Holevo73} where we only need to optimise over the free parameters $p(m)$, ${\rho}_m$ and ${\mathcal{A}}$. 

This implies one does not need to optimise over Bob's operation to obtain the Holevo quantity for the process. 
\end{proof}

A direct consequence of this theorem is that $\chi(W)^{A \rightarrow{B}} \leq \log(d_{B_I})$, because we have reduced the Holevo quantity of a process to that of an ensemble of states in $B_I$. In turn, this allows us to establish a bound on the product encoding capacities, i.e., $C^{(1)}(W)^{A\rightarrow{B}}$ and $C(W)^{A\rightarrow{B}}$, according to Eq.~\eqref{Eq:one-shot} and Eq.~\eqref{Eq:product_enc_joint_decod}, respectively. However, for joint encoding schemes we need to evaluate the regularised Holevo quantity for the optimum channel $\mathcal{N}({\mathcal{A}},{\mathcal{B}},W)$, as shown in Eq.~\eqref{jointencoding}. In ligth of this, we introduce the following theorem that bounds the capacity, $C^E(W)^{A\rightarrow{B}}$.

\begin{theorem}
 The joint encoding capacity for a process, $C^E(W)^{A\rightarrow{B}}$, is limited to Alice's joint CPTP map ${\mathcal{A}}_m^{(n)}{=}{\rho}_m^{(n)}{*}{\mathcal{A}}^{\otimes n}$ with $n\rightarrow{\infty}$, and the distribution $p(m)$. 
 \begin{align}
    C^E(W)^{A\rightarrow{B}} &{=} \nonumber \\ \lim_{n\rightarrow{\infty}}\underset{{\rho}_m^{(n)},{\mathcal{A}}, p(m)}{{\mathrm{max}}}
    \ \frac{1}{n} &\bigg(S\left[\sum p(m) {\Gamma} ({\mathcal{A}}_m^{(n)}, W )\right] - \nonumber \\
    &\sum p(m)S\left[{\Gamma} ({\mathcal{A}}_m^{(n)}, W)\right]\bigg).
    \label{W_Holevo_joint}
\end{align}
Here the map $\Gamma$ transforms the joint CPTP map $A_m^{(n)}$ and the process $W$ to an entangled state at Bob's input space $\otimes_{j=1}^n B_I^j$.
 \end{theorem}
\begin{proof}

First, we apply Lemma ~\ref{noBO} to each copy of the channel $\mathcal{N}(\mathcal{A},\mathcal{B},W)$ and replace it with $\mathcal{N}(\mathcal{A},\mathcal{B},\, _{B^j_O}W)$. Then, in a joint encoding scheme, we apply $\otimes_j\mathcal{N}(\mathcal{A},\mathcal{B},\, _{B^j_O}W)$ to a (possibly entangled) joint state ${\rho}_m^{(n)}$. Combining this joint state with the $n$ copies of Alice's operation $\mathcal{A}$, we obtain a joint CPTP map $\mathcal{A}^{(n)}_m :  {\otimes}_j A^j_I \rightarrow {\otimes}_j A^j_O$, with Choi representation $A^{(n)}_m = \rho_m^{(n)}{*}A^{\otimes n}$. Plugging $A^{(n)}_m$ into the $n$ copies $_{B^j_O}W$, we get $A^{(n)}_m{*}\left({\otimes_j} _{B^j_O}W\right) = \mathbb{1}^{\otimes_{j=1}^n B_O^j}\otimes {\Gamma} ({\mathcal{A}}_m^{(n)}, W)$, where ${\Gamma} ({\mathcal{A}}_m^{(n)}, W)$ ${\in}\otimes_{j=1}^n B_I^j$ is a (possibly entangled) state, defined as 
\begin{align}
{\Gamma} ({\mathcal{A}}_m^{(n)}, W) =\frac{{A}_m^{(n)}*\big(\tr_{\otimes_{j=1}^n B_O^j}W^{\otimes{n}}\big)}{\Pi_{j=1}^n d_{B_O^j}}.
\end{align}
One can extend this setup to $n{\rightarrow}{\infty}$ and achieve a joint state at Bob's input Hilbert space $\otimes_{j=1}^\infty B_I^j$. Similar to Theorem ~\ref{Theorem:2}, we calculate the maximum amount of classical information encoded in the ensemble $\{p(m), {\Gamma} ({\mathcal{A}}_m^{(n)}, W )\}$ and regularise it to obtain the joint encoding capacity $C^E(W)^{A\rightarrow{B}}$ where the free parameters are of course, $p(m)$, ${\rho}_m^{(n)}$ and ${\mathcal{A}}$. Thus we obtain Eq.~\eqref{W_Holevo_joint}.
\end{proof}

\begin{corollary}
 The capacity $C^E(W)^{A\rightarrow{B}}$ is upper bounded by the logarithm of the dimension of Bob's input Hilbert space, i.e. $C^E(W)^{A\rightarrow{B}} \leq \log (d_{B_I})$.
\end{corollary}

\begin{proof} This is the consequence of Holevo's theorem \cite{Holevo73}. The information content of the ensemble $\{p(m), \Gamma ({A}_m^{(n)}, W)\}$ cannot exceed logarithm of the dimension of $\Gamma ({A}_m^{(n)},W)$, i.e., $n\log (d_{B_I})$. Regularising this quantity proves the corollary. 
\end{proof}

Now we summarise our results. If we consider Alice's input message $m$ and Bob's output message $m'$, we can introduce a chain of inequalities for different classical capacities of the process $W$.

\begin{align} 
    I(m':m) &\leq C^{(1)}(W)^{A \rightarrow{B}} \leq C(W)^{A \rightarrow{B}} = \chi(W)^{A \rightarrow{B}} \nonumber \\
   & \leq C^E(W)^{A \rightarrow{B}}  \leq \log (d_{B_I}).
   \label{A_to_B}
\end{align}
One can write down a similar chain of inequalities for a communication protocol from Bob to Alice. Note that this inequality holds even for a process $W$ that contains shared entanglement between Alice's and Bob's input Hilbert spaces. This does not contradict the higher capacity achievable in an entanglement assisted communication protocol, such as super-dense coding \cite{superdense-coding}, because, when applying the inequalities in Eq.~\eqref{A_to_B}, one has to consider the total dimension of Bob's input Hilbert space, which consists of the part of the shared entangled state in Bob's possession and the quantum state that Alice communicates to him.

\section{Broadcast communication}\label{sec:5}
Having established the notion of one-way communication through a process, we proceed to explore scenarios where all communicating parties can transmit and receive information. 

\emph{Two-party communication:} Let us first consider the two party situation, where Alice (Bob) sends the message $m (k)$ and Bob (Alice) receives the message $m'(k')$. The possibility to violate causal inequalities indicates that indefinite causal order can indeed provide an advantage in some two party games \cite{oreshkov12, Branciard2016}; however, it is unclear if this advantage results in a communication enhancement. To address this question, it is necessary to find limits on two-way communication for causally separable processes. In this section we address this question.

There are at least two ways a process can be used as a resource for bidirectional communication, depending on whether Alice's and Bob's instruments are fixed or if they are chosen depending on the direction of communication attempted. In the first case, the parties produce a single probability distribution $P(m',k'|m,k)$ from the process, and one looks for communication in the marginals $P(m'|m)$, $P(k'|k)$. In the second case, the parties can generate different probability distributions depending on who is sending and who is receiving. The one-directional capacities for the first case are upper bounded by those in the second case, as the best instrument to receive a message might differ from the best to send a message. In this section, we will be mostly concerned with the second case.

Let us then consider a scenario where the order between Alice's and Bob's local operations is determined based on a random outcome. We represent a process where Alice can signal to Bob, but not the other way around, by $W^{A \prec B} =\, _{B_O}W^{A \prec B}$ and the reversed direction of signalling by $W^{B \prec A} {=}\, _{A_O}W^{B \prec A}$. The process matrix $W_{\mathrm{Sep}}$ in this case is a convex combination of $W^{B \prec A}$ and $W^{A \prec B}$ \cite{araujo15}: 
\begin{align}
 W_{\mathrm{sep}} = \lambda W^{B \prec A}+(1-\lambda)W^{A\prec B},
 \label{mixed}
\end{align}
where $0\leq \lambda \leq 1$ is the probability for Bob to be first.
We call such a process a \emph{causally separable process} \cite{oreshkov12}. We investigate a scenario where both Alice and Bob are trying to send information to each other through the background process $W_{\mathrm{Sep}}$. A reasonable attempt to quantify this bi-directional communication is to evaluate the sum of two one-shot capacities, $C^{(1)}(W)^{A\rightarrow{B}}$ and $C^{(1)}(W)^{B\rightarrow{A}}$.  We investigate this quantity and evaluate an operationally significant upper bound from the perspective of the classical capacity of the process.

\begin{theorem}
For a bi-directional communication protocol through a causally separable process, defined in Eq.~\eqref{mixed}, the following inequality holds: 

\begin{align}
&C^{(1)}(W_\mathrm{sep})^{A{\rightarrow}{B}}{+} C^{(1)}(W_\mathrm{sep})^{B{\rightarrow}{A}} \nonumber \\
&{\leq} \lambda \log {(d_{A_I})} {+} (1{-}\lambda) \log{(d_{B_I})}.   
\label{Eq:Bi-direction}
\end{align}
\end{theorem}

\begin{proof}
Considering a fixed input probability distribution $P(a)$, the following linear relationship among marginal conditional probabilities holds for a causally separable process \cite{oreshkov12}.
\begin{align}
    P(a'|a)_{W_\mathrm{sep}} &= \lambda P(a'|a)_{W^{B \prec A}} + (1-\lambda) P(a'|a)_{W^{A \prec B}}.
\end{align}
With $a{\in}\{m,k\}$ and $a'{\in}\{m',k'\}$. Let us consider $A{\rightarrow}{B}$ communication. Consequently we can write:
\begin{align}
    &C^{(1)}(W_{Sep})^{A\rightarrow{B}}=\max I(m':m) \nonumber \\
    &\leq \max \left[\lambda I(m'{:}m)_{W^{B \prec A}} + (1-\lambda)  I(m'{:}m)_{W^{A \prec B}}\right]\nonumber \\
    &=(1-\lambda) \max I(m'{:}m)_{W^{A \prec B}} \nonumber \\
    &{\leq}(1-\lambda)\log(d_{B_I}).
    \label{Eq:A_to_B_bound}
\end{align}
Here, the first equation is due to Eqs.~\eqref{accessible_information}, and  ~\eqref{Eq:Capacity_w}. The maximisation is taken over Alice's and Bob's local operations, their message ensembles and the POVM operations. The first inequality is due to the fact that mutual information $I(a'{:}a)$ is a convex function of $p(a'|a)$ for a fixed input probability distribution $p(a)$ \cite{Cover06}. We obtain the second equality because, for a definite ordered scenario $B\prec A$, output $m'$ of Bob's local lab becomes independent of Alice's input. This makes $I(m{:}m')_{W^{B \prec A}} = 0$. The final inequality is due to Eq.~\eqref{A_to_B}. We apply a similar set of reasoning to obtain a bound for $B\rightarrow{A}$ communication to obtain, 
\begin{align}
    &C^{(1)}(W_{Sep})^{B\rightarrow{A}} \leq \lambda\log(d_{A_I}).
    \label{Eq:B_to_A_bound}
\end{align}
Adding  Eq.~\eqref{Eq:A_to_B_bound} and \eqref{Eq:B_to_A_bound}, we find

\begin{align}
C^{(1)}(W_\mathrm{sep})^{A{\rightarrow}{B}}&{+} C^{(1)}(W_\mathrm{sep})^{B{\rightarrow}{A}} \nonumber \\
& \leq \lambda \log(d_{A_I}) + (1-\lambda) \log(d_{B_I}).
\label{bipartitebound}
\end{align}

\end{proof}

For the particular case $d_{A_I}{=}d_{B_I}{=}d$, we see that the sum of two one-shot capacities is upper bounded by $\log (d)$. In other words, the total communication in causally separable processes can be no more than maximal one-way communication. We note that a weaker version of this inequality holds for the scenario where the parties' instruments are fixed regardless of the attempted direction of communication. In this case, the single-shot capacities coincide with the mutual information obtained from a single conditional probability distribution $P(m',k'|m,k)$, resulting in the inequality $I(m'{:}m) + I(k'{:}k) {\leq} \log (d)$. This is an example of an \emph{entropic causal inequality}, first considered in Ref.~\cite{Miklin_2017}. Remarkably, no violation of this inequality is known, and our own numerical search also did not reveal any violation of Eq.~\eqref{Eq:Bi-direction}. This suggests that the bound on the total bidirectional communication might hold for general processes. 

We note that the bound we established applies to all the quantum processes for which a physical interpretation is known. For example, in a process with coherent control of causal order, such as the quantum switch \cite{chiribella09b}, tracing out the control leads to a separable bipartite process, to which the bound applies. More generally, it has been shown that any bipartite processes that admit a unitary extension is causally separable \cite{Barrett2020, Yokojima2020}.

 \emph{Multi-party communication:} The above-mentioned protocol can be extended to multiple parties. In that case, each party tries to communicate his/her information to the remaining parties. Similarly to above, we consider a process for $N$ parties, $A^{(1)}, A^{(2)}, ... , A^{(N)}$ that can be written as a probabilistic mixture of permutations of different causal order:
\begin{align}
    W^N_\mathrm{sep} = \sum _\sigma {q_{\sigma}} W^{\sigma}. 
\end{align}
Here, $\sigma$ denotes the different permutations of the communicating parties and $q_{\sigma}$ denotes the probability of occurrence of each permutation. Although this is not the most general process with definite causal order \cite{oreshkov15, Wechs2019}, it is one of particular interest, as it represents a scenario where the order among parties can be set by external, random, variables, but is independent of the parties' actions.

Motivated by the previous section, we intend to find an upper bound to the quantity $\sum _{i,j}C^{(1)}(W^N_\mathrm{sep})^{i{\rightarrow}{j}}$.  Here $i{\rightarrow}{j}$ refers to signalling from the party $A^{(i)}$ to the party $A^{(j)}$.

\begin{theorem} 
If dimensions of all the input Hilbert spaces of the communicating parties are equal ($d$), then  \begin{align}
 \sum _{i,j}C^{(1)}(W^N_\mathrm{sep})^{i{\rightarrow}{j}} {\leq} \frac{N(N-1)}{2} \log (d).   
\end{align}
\end{theorem}

\begin{proof}  we can write the conditional probability $P(\vec{m}'|\vec{m}) = \sum _\sigma {q_{\sigma} P_\sigma (\vec{m}'|\vec{m})}$. We can write down the marginals $P(m'_j|m_i) = \sum _\sigma {q_{\sigma} P_\sigma (m'_j|m_i)}$ $\forall {i,j}$. By the convexity of mutual information and the inequalities introduced in Eq.~\eqref{A_to_B}:
\begin{align} \label{networkinfo}
 C^{(1)}(W^N_\mathrm{sep})^{i{\rightarrow}{j}} \leq \mathrm{max}&\sum _\sigma {q_{\sigma} I_\sigma (m'_j:m_i)}\nonumber \\
 = \mathrm{max}&\sum _{\forall{\{i,j\}}\,{|}\,\sigma(i) \prec \sigma(j)} {q_{\sigma} I_\sigma (m'_j:m_i)} \nonumber \\
 \leq &\sum _{\forall{\{i,j\}}\,{|}\,\sigma(i) \prec \sigma(j)} {q_{\sigma} \log (d_{A^{j}_I})} \nonumber \\
= &\sum _{\forall{\{i,j\}}\,{|}\,\sigma(i) \prec \sigma(j)} {q_{\sigma} \log (d)}
\end{align}
The maximisation is taken over all communicating parties' local operations, their message ensembles and the POVM operations. The first inequality is due to the convexity of mutual information relative to mixtures of conditional probabilities (as in the bipartite case). The first equality follows from the fact that if the party $\sigma(j)$ is in the causal past of the party $\sigma(i)$, then $I(m'_j{:}m_i) = 0$. $d_{A^{j}_I}$ is the dimension of the input Hilbert space of the party $A^{j}$. The second equality follows because of our assumption of all the dimensions of the input Hilbert spaces being equal.  Now, it is easy to see that the $n$-th party has total $n{-}1$ parties in his/her causal past. Therefore, considering each party trying to communicate with the remaining $N{-}1$ parties, the total number of available channels are $N(N{-}1) - \sum _{n{=}1}^{N} (n{-}1) = N(N{-}1)/2$. This results in
\begin{align}
    \sum _{i,j}C^{(1)}(W^N_\mathrm{sep})^{i{\rightarrow}{j}} \leq \frac{N(N-1)}{2} \log (d).
\end{align}
\end{proof}

The key property that leads to the above bounds is the convexity of mutual information under probabilistic mixtures of classical channels. With this in mind, we see that the above results can be extended directly to the product encoding, joint decoding setting, replacing the one-shot capacity $C^{(1)}$ with the HSW capacity $C$. Indeed, we have seen that $C$ is given by the (maximised) Holevo quantity $\chi$ of the one-way channel generated by a process and, just like mutual information, $\chi$ is convex over the probabilistic mixture of channels. It remains an open question whether higher total transmission rates can be achieved in a joint encoding setting.

\section{Conclusion}\label{sec:conclusion}

We have formalised classical communication through a general quantum causal structure, which may be probabilistic or indefinite. We have defined the Holevo quantity as well as different classical capacities for an arbitrary process and established relationships among them. We have found that, for one-way communication, the various capacities can be reduced to those of ordinary channels, up to an optimisation over the operations performed in local laboratories. We have further shown that, for one-way communication, the classical capacity of a process cannot exceed the Holevo bound---at most one classical bit per received qubit---even in case of indefinite causal order. 

Next, we have quantified bi-directional and more generally broadcast communication protocols for processes with definite but classically uncertain causal order. We have demonstrated that the total amount of communication between two parties cannot exceed the maximal one-way capacity in a fixed causal order, with a similar bound extending to multipartite broadcast communication. One can ask whether a process with an indefinite causal structure can violate these bounds. We have answered this negatively for coherent control of causal order, as in the quantum switch \cite{chiribella09b}. It is an open question whether a more general process can violate the bounds. As we have not found any violation, it is an interesting possibility that the bounds we have found might constitute a universal limit to the total communication possible in any process.

\section*{Acknowledgement} KG thanks the organisers of the Quantum Information Structure of Spacetime 2020 workshop for giving an opportunity to present this talk. We are grateful to Alastair Abbott, Ding Jia, Nitica Sakharwade, Marco Tomamichel, and Magdalena Zych and for helpful discussions.  This work has been supported by the Australian Research Council (ARC) by Centre of Excellence for Engineered Quantum Systems (EQUS, CE170100009), F.C.\ acknowledges support through an Australian Research Council Discovery Early Career Researcher Award (DE170100712). KG is supported by the RTP scholarship from the University of Queensland. We acknowledge the traditional owners of the land on which the University of Queensland is situated, the Turrbal and Jagera people.

\bibliography{Bidirection.bib}

\end{document}